\documentclass[runningheads]{llncs}

\usepackage{xspace}
\usepackage{subfiles}
\usepackage{amsfonts}
\usepackage{amsmath,amssymb,amsfonts}
\usepackage{multirow,array}
\usepackage{subfig}
\usepackage{booktabs}
\usepackage[T1]{fontenc}
\usepackage{graphicx}
\usepackage[boxruled]{algorithm2e}
\usepackage{dsfont}
\usepackage[dvipsnames]{xcolor}
\usepackage{mathtools}
\usepackage[normalem]{ulem}

\long\def\todo#1 { {\bf TODO:} [{\color{gray} #1}] }

\newcommand{\eg}{{\em e.g.}\xspace}
\newcommand{\ie}{{\em i.e.}\xspace}

\newcommand{\Ph}{\mathcal{P}_h}
\newcommand{\Pm}{\mathcal{P}_m}

\newcommand{\Sh}{\texttt{S}_h}
\newcommand{\Sm}{\texttt{S}_m}
\newcommand{\Smp}{\texttt{S}'_{m}}

\newcommand{\Rh}{\mathcal{R}_{h}}

\newcommand{\Rm}{\mathcal{R}_m}

\newcommand{\Rhi}{\mathcal{R}_{h,i}}
\newcommand{\Rdi}{\mathcal{R}_{d,i}}
\newcommand{\Rdpi}{\mathcal{R}_{d',i}}
\newcommand{\Rmi}{\mathcal{R}_{m,i}}
\newcommand{\Dm}{\mathcal{D}_m}
\newcommand{\Di}{\mathcal{D}_i}
\newcommand{\Oracle}{\mathcal{O}}
\newcommand{\order}{order}
\newcommand{\N}{\mathcal{N}}
\newcommand{\Nh}{\mathcal{N}_h}
\newcommand{\Nm}{\mathcal{N}_m}
\newcommand{\Tn}{T_{now}}
\newcommand{\Te}{T_e}
\newcommand{\vi}{v_i}
\newcommand{\Vh}{V_h}
\newcommand{\Vm}{V_m}
\newcommand{\Vmp}{V'_m}

\newcommand{\threshold}{t}

\newcommand{\BSC}{\textit{Bribery Smart Contract}\xspace}
\newcommand{\bsc}{bribery smart contract\xspace}
\newcommand{\Bsc}{Bribery smart contract\xspace}
\newcommand{\DS}{Deposit-Slashing\xspace}
\newcommand{\ds}{deposit-slashing\xspace}
\newcommand{\PoW}{Proof-of-Work\xspace}
\newcommand{\PoS}{Proof-of-Stake\xspace}
\newcommand{\NE}{Nash Equilibrium\xspace}
\newcommand{\Ne}{Nash equilibrium\xspace}
\newcommand{\power}{power\xspace}
\newcommand{\powert}{power threshold\xspace}
\newcommand{\cns}{nodes\xspace}
\newcommand{\cn}{node\xspace}
\newcommand{\magnate}{magnate\xspace}

\newcommand{\minion}{minion\xspace}
\newcommand{\minions}{minions\xspace}
\newcommand{\bc}{block\-chain\xspace}

\newcommand{\Sys}{\ensuremath{S}\xspace}
\newcommand{\Sprime}{\ensuremath{S'}}

\newcommand{\utility}[3]{
\ensuremath{\text{Utility}^{#1}_{#2}(#3)}}

\begin{document}
%


\title{Breaking Blockchain Rationality \\
with Out-of-Band Collusion}
%
%

\author{Haoqian Zhang\inst{1} \and
Mahsa Bastankhah\inst{1} \and
Louis-Henri Merino\inst{1} \and \\
Vero Estrada-Galiñanes\inst{1} \and  
Bryan Ford\inst{1}}
\authorrunning{Haoqian Zhang et al.}
%
\institute{École polytechnique fédérale de Lausanne(EPFL)
\email{{haoqian.zhang,mahsa.bastankhah,louis-henri.merino,vero.estrada,bryan.ford}@epfl.ch}}
\maketitle              

\begin{abstract}

Blockchain systems often rely on rationality assumptions for their security,
expecting that \cns are motivated to 
maximize their profits.
These systems thus design their protocols to
incentivize \cns to execute the honest protocol
but fail to consider out-of-band collusion.
Existing works analyzing rationality assumptions
are limited in their scope, either by focusing
on a specific protocol or relying on non-existing
financial instruments.
We propose a general rational attack on rationality by 
leveraging an external channel 
that incentivizes \cns to collude against the honest protocol.
Our approach involves
an attacker creating an out-of-band \bsc
to motivate \cns to double-spend their transactions
in exchange for shares in the attacker's profits.
We provide a game theory model to prove that
any rational \cn is incentivized to follow the malicious protocol.
We discuss our approach to attacking 
the Bitcoin and Ethereum blockchains,
demonstrating that irrational behavior can be rational
in real-world blockchain systems
when analyzing rationality in a larger ecosystem.
We conclude that rational assumptions only appear
to make the system more secure
and offer a false sense of security
under the flawed analysis.

\end{abstract}
\section{Introduction}


Blockchain systems often rely on rationality assumptions
to ensure their security by providing financial incentives
for adhering to the honest protocol.
For example, in Proof-of-Work,
miners are incentivized to work on the longest chain
as it increases their expected chances of 
having their blocks accepted in the blockchain.
Similarly, in Proof-of-Stake, 
such as the one recently adopted by Ethereum~\cite{ethereumMerge},
validators are disincentivized from malicious behavior,
such as signing two blocks with the same height, 
due to the loss of part of their deposits.
These incentive mechanisms seem to secure these systems
as any entity deviating from the honest protocol would have a lower or negative expected return.

However, as many previous works 
demonstrated~\cite{eyal2018majority,liao2017incentivizing,Babaioff2012,Houy2014,Eyal2015,Know2017},
those mechanisms might not be incentive-compatible,
\ie, there exists a more profitable alternative strategy 
that deviates from the honest protocol.
For instance,
selfish mining is a strategy to increase miners' expected return
by deviating from the longest-chain rule 
expected by the Bitcoin mining protocol~\cite{eyal2018majority}.
Whale attacks incentivize
miners to fork the chain
to include an off-the-blockchain transaction
with a substantial transaction fee~\cite{liao2017incentivizing}.

Whereas those previous works focus on specific protocols
within individual blockchain systems,
we question the incentive mechanism at a meta-level:
Are those blockchain systems rely on rationality assumptions secure in general?
We try to answer this research question by
considering attacks beyond their ecosystem
taking into account the broader influences of the outside world on the system.
What is considered irrational behavior within their ecosystem
might be rational when analyzing rationality in the context
of a larger ecosystem.

We demonstrate that rationality
assumptions can be defeated by attacks driven by rationality. 
Specifically, an attacker creates an out-of-band \bsc
that incentivizes \cns to double-spend the attacker's transactions.
In return, the attacker can then share the profits from the double-spending
with colluded consensus nodes,
offering a financial incentive for them to commit the attack 
in the first place.

A closely related work by Ford and Böhme~\cite{ford2019rationality}
also offer a general rational attack on rationality.
However, their attack method relies on financial instruments
that are either non-existent or not well-established
in the cryptocurrency markets.
We, on the other hand, 
eliminate the need for non-existent financial instruments
and significantly relaxes the requirements to launch the attack.

To prove that out-of-band collusion breaks blockchain systems' 
rationality assumptions,
we propose a game theory model and use it to analyze a blockchain
system before and after launching our attack.
We find that in the absence of the attack,
following the honest protocol is a strict \Ne 
that discourages nodes from deviating;
however, in the presence of our attack, 
the honest protocol becomes a weakly dominated strategy. 
In particular, 
we identify a finite sequence of deviations from 
the honest protocol where each deviating \cn
obtains at least the same reward as before the deviation. 
This sequence ultimately leads to a state where all the \cns follow our attack. 
Furthermore, we prove that following our attack is a strict \Ne,
thus disincentivizing further deviation.

We provide an outline of the steps required
to break the longest-chain rule in Bitcoin
and the \ds protocol in Ethereum.
Our work implies that 
rationality assumptions only appear
to make the system more secure
and provide a false sense of security.

\section{Assumptions Underlying the Attack}
\label{sec:assumptions}

This section introduces the following assumptions for our attack model:

\subsubsection{Assumption 1:}
\label{Assumption_blockchain}
We consider the target system \Sys{} to 
be an open financial payment network operating
on blockchain rails, where any client can initiate a transaction.
\Sys{} is maintained by a set of 
rational \cns $\N=\{1,2,\dots,n\}$
who seek to maximize their profits.
We assume that each \cn, $i \in \N$, 
has the \power of $\vi$,
\ie, the voting \power to decide the next block in the blockchain system.
For example, the voting \power in a \PoW blockchain is the 
\cns' computational power and the voting \power
in a \PoS blockchain is \cns' stake amount,
whereas the voting power in a practical Byzantine Fault Tolerance(PBFT) blockchain 
is the existence of an approved node.
We normalize the \power distribution such that the sum of all the nodes' power is equal to 1: $\sum_{i=1}^{n} \vi = 1$.
For simplicity, 
we assume that the number of \cns and their \power distribution 
remains constant;
however, our model also applies to 
the dynamic number of \cns with smooth \power changes.

\subsubsection{Assumption 2:}
\label{Assumption_outside}
We assume the existence of an open system \Sprime{}
that supports smart contracts and 
has access to a perfect oracle mechanism $\Oracle$
that can access real-time state information on \Sys{}
without manipulation.
To avoid \Sprime{} and $\Oracle$ being attacked by the same rational attack,
we assume that \Sprime{} and $\Oracle$ do not rely on any rationality assumption,
and their security assumptions hold. 
For example,
\Sprime{} could be a PBFT-styled blockchain,
where at most $f$ of $3f+1$ \cns can fail or misbehave,
and $\Oracle$ can solely rely on trusted hardware~\cite{costan2016intel} to provide truthful information from \Sys{}.

\subsubsection{Assumption 3:}
\label{Assumption_honest}
The system \Sys{} leverages, in some fashion, rationality assumptions
to incentivize nodes to follow the \Sys-defined honest protocol $\Ph$.
Mathematically, we assume there is a well-known 
\powert $\threshold$ such that,
within a time period,
if $\Nh \subset \N$ with 
$\sum_{i \in \Nh} \vi > \threshold$ follows the honest protocol $\Ph$,
for $i \in \Nh$ expects to receive a reward of $\Rhi > 0$,
and for $i \notin \Nh$ expects to obtain a reward $\Rdi$.
We assume that $\forall i \in \N, \Rdi<\Rhi$.
$\Rdi$ can be negative, 
\ie, a \cn receives punishment for deviating from $\Ph$.

\subsubsection{Assumption 4:}
\label{Assumption_malicious}
We assume the existence of a malicious protocol $\Pm$ 
that differs from the expected behavior such that,
within the same time period,
if $\Nm \subset \N$ with 
$\sum_{i \in \Nm} \vi > \threshold$ follows the malicious protocol $\Pm$,
for $i \in \Nm$ can expect to receive a reward of $\Rmi$,
and for $i \notin \Nm$ can expect to obtain a reward of $\Rdpi$.
We assume that $\forall i \in \N, \Rdpi<\Rmi$ and $\Rmi>\Rhi$
as the malicious protocol is only worthwhile for attackers if it provides them with greater rewards.
In Section \ref{subsec:pmal},
we show that there always exists a malicious protocol 
capable of double-spending attacks to satisfy this assumption
in real-world blockchain systems.

\subsubsection{Assumption 5:}
\label{Assumption_practical}
We assume that the underlying consensus requires $t \geq \frac{1}{2}$
to avoid \cns split into two independent functional subsets.
We also assume that no single \cn can abuse the system, meaning
that
$\forall i \in \N, \vi<\threshold$.
For simplicity,
we assume that if neither $\Ph$ nor $\Pm$ has enough \cns to execute,
\Sys{} loses liveness, and nobody gets any reward.
\section{Rational Attack on Rationality}
\label{sec:attack}

This section presents our attack on rationality at a high-level.
We begin by demonstrating that
no rational \cn would execute $\Pm$ without collusion.
We then introduce an attacker who creates a \BSC
on \Sprime{} that incentivizes the \cns on \Sys{} to launch the attack.

\subsubsection{Without Collusion:}
In the absence of collusion between \cns,
each \cn is incentivized to follow the honest protocol $\Ph$;
%
no single rational \cn will deviate from $\Ph$ as the expected reward
is lower than that of following $\Ph$
($\Rdi<\Rmi$ in Assumption 3).
Therefore, when there is no collusion,
\Sys{} is secure under the rational assumption
(we present a game theory analysis in Section \ref{subsec:game0}).
However, one cannot optimistically assume 
that such collusion will not exist.

\subsubsection{Magnate-Coordinated Collusion:}

\begin{algorithm}[t!]
\caption{\Bsc to incentivize collusion}
\label{alg:smartcontract}

\SetCommentSty{textnormal}
\DontPrintSemicolon
\SetKwInOut{Input}{input}
\SetKwFor{For}{for}{}{end}
\SetKwProg{Init}{Init}{:}{}
\SetKwProg{Commit}{Commit}{:}{}
\SetKwProg{Distribute}{Distribute}{:}{}
\SetKwProg{Attack}{Attack}{:}{}

\BlankLine
\Init{Upon creating the bribery smart contract}{
    Set $\Te$ as the expiration time\; 
    Set $\Pm$ as the malicious protocol\;
    Deposit $\Dm$ by the \magnate\;
    $\Nm \gets \varnothing$\;
    $\order \gets \Ph$
}
\;
\Commit{Upon receiving \cn $i$'s commitment request}{
    $\Nm \gets \Nm \cup i$\;
    Deposit $\Di$ by $i$\;
}
\;
\Attack{Upon $\sum_{i\in\Nm}\vi>\threshold$}{
    $\order \gets \Pm$
}
\;
\Distribute{Upon receiving the request from $i \in \Nm$ for the first time}{
    \If{Attack is successful and $i$ has executed $\Pm$}{
        Distribute $\vi\Dm+\Di$ to $i$
    }
    \If{Attack is not successful and $\Tn>\Te$}{
        Distribute $\Di$ to $i$
    }
}
\end{algorithm}

When an \Sprime{} exists, 
an attacker (referred to as a \textit{\magnate}) 
can use it to coordinate collusion between nodes (Assumption 2).
To defeat \Sys{},
the \magnate can create a \bsc
to attract \cns (referred to as \textit{\minions} and denoted by $\Nm$).

We use the double spending attack induced by the \magnate 
as an example to illustrate a possible malicious protocol $\Pm$.
The \magnate needs to use a \bsc to specify the transaction to be reverted,
and order \minions to work on a fork 
that allows the \magnate to double-spend the transaction.
To ensure the attack's success,
the \magnate must guarantee that
each \cn can expect a higher reward, 
\ie, $\Rm > \Rh$.
In the case of this double-spending attack, 
each \cn can still expect to receive the rewards
that a \cn executing $\Ph$ would typically get,
such as block rewards and transaction fees.
However, \cns can now expect to receive
a share of the profits obtained by the \magnate
through double-spending by having the \cns execute $\Pm$.
Therefore, the \magnate has successfully produced
a reward $\Rm$ strictly greater than $\Rh$.
Note that 
the double spending attack is just one example of a malicious protocol. 
As long as the malicious protocol $\Pm$ produces a higher reward, \ie, $\Rm > \Rh$,
it works in our model to defeat rationality.

We outline the design of the \bsc (Algorithm \ref{alg:smartcontract}) 
on \Sprime{} that would enable the \magnate to execute the attack successfully.
All parties must be held accountable if any party defects
to ensure a successful attack in practice.
During the creation of the smart contract,
the \magnate thus deposits $\Dm$ to be shared
among the \cns if the attack is successful.
In addition, when joining the \bsc, each \minion is required to 
deposit $\Di$ to be slashed in case of a defect.
When the \minions' total voting \power exceeds $\threshold$,
the \bsc orders them to execute $\Pm$.
%
The smart contract then can monitor the attack
through the oracle $\Oracle$ (Assumption 2) and upon success, 
returns the deposits with a share of $\Dm$ 
to each \minion.
If the \magnate fails to attract enough \minions to commit the attack,
the deposits are still returned to each \minion
after an expiration time,
making the commitment of the attack by a node risk-free.
The \magnate can also require a large $\Di$
as each colluded \cn expects to get back $\Di$ eventually
(we discuss how to choose $\Di$ in Section \ref{subsec:game1}).
However, 
if a \minion does not follow the order from the \bsc, their deposit is burned, 
thus incentivizing each \minion to follow the order.

Given the \bsc, 
a rational \cn is incentivized to commit and execute $\Pm$,
as, intuitively, every \cn can benefit.
If a \cn does not participate in the attack,
it can, at most, obtain $\Rh$.
However, if a \cn joins the attack,
it will receive at least $\Rh$
with the opportunity of increasing its reward to $\Rm$.
We offer a game theory analysis on \cn collusion in Section \ref{subsec:game1}.
We emphasize that, in this attack,
the \magnate does not even need to control any part of \Sys{} or \Sprime{},
making such an attack doable with a low barrier to launch.
\section{Game Theoretic Analysis}\label{sec:gametheory}

In this section, 
we formalize the behavior of \Sys{} \cns
and examine the possibility of deviation
first without any collusion
and then with collusion through the \bsc on \Sprime.
%
%

In the absence of collusion, following the honest protocol
$\Ph$ is a strict \Ne,
meaning that no player will deviate as deviation leads to a lower payoff.
However, in the presence of the \bsc, 
following the protocol $\Ph$ is a weakly dominated strategy 
and thus is no longer a strict \Ne. 
In particular, we identify a sequence of deviations from $\Ph$ 
where each deviant node obtains at least the same payoff as before.
We show that this sequence of deviations ends with following the \bsc orders.
Furthermore, we prove that following the \bsc orders is a strict \Ne,
yielding the maximum payoff of the game.
As a result, no rational player would deviate from it.
%

Additionally, we provide a bound on the amount of money that minions
should deposit in the \bsc to ensure that they do not deviate from the
\bsc's commands.

\subsection{Game 0: Without Collusion}
\label{subsec:game0}
We model the behavior of the \cns in the absence 
of any external factors as a strategic-form game 
$\Gamma_{0} = (\mathcal{N} , \{\Sh , \Sm\}^n , \utility{0}{i}{.}_{i \in \mathcal{N}})$.
$\mathcal{N}=\{1,2,\dots,n\}$ is the set of nodes (players) of the game. Each
node $i$ has \power $\vi$ such that $\sum_{i \in \mathcal{N}} \vi = 1$.
Each player can choose the honest strategy $\Sh$ (corresponding with the protocol $\Ph$) or the malicious strategy $\Sm$ (corresponding with the protocol $\Pm$). 
We denote the chosen strategy of node $i$ by $s_i$.

We define $\Vh$ as the total \power of the nodes that choose strategy $\Sh$ and 
$\Vm$ as the total \power of the nodes which follow $\Sm$, \ie, 
$$\Vh \coloneqq \sum_{i\in \mathcal{N}} \vi 1_{\{s_i=\Sh\}}$$
$$\Vm \coloneqq  \sum_{i\in \mathcal{N}} \vi 1_{\{s_i=\Sm\}} = 1 - \Vh \text{.}$$

Finally, we define the utility function of node $i$, $\utility{0}{i}{.}$, which is a function of $i$'s and other players' strategies as follows:

 \[
    \utility{0}{i}{s_1,\dots,s_n} =
  \begin{cases}
         {\Rhi} & \text{If } s_i= \Sh  \quad \& \quad \Vh > \threshold\\
          {\Rdpi}   & \text{If } s_i= \Sh \quad \& \quad \Vm > \threshold\\
         {\Rdi} & \text{If } s_i= \Sm \quad \& \quad \Vh > \threshold\\
         {\Rmi} & \text{If } s_i= \Sm \quad \& \quad \Vm > \threshold\\
         0 & \text{If } \Vh , \Vm \leq \threshold
        \end{cases}      
  \]
  $$\text{with } \Rhi > \Rdi, \Rmi > \Rhi > 0,\Rmi > \Rdpi \text{.}$$
  Suppose $\Vh > \threshold$, 
  \ie, majority \power is dedicated to the strategy $\Sh$, 
  player $i$ obtains reward $\Rhi$ by following $\Sh$ 
  and obtains $\Rdi$ otherwise. 
 %
 Similarly, when the majority adopts $\Sm$, 
 player $i$ obtains reward $\Rmi$ by following $\Sm$ 
 and gets $\Rdpi$ otherwise. 
 We assume that $\Rmi > \Rhi$ (Assumption 4).
 If both $\Vh$ and $\Vm$ are smaller than $\threshold$, 
 all the nodes receive a payoff of $0$ (Assumption 5).

\begin{theorem}\label{theorem:game0}
    In the strategic-form game $\Gamma_0$ if  $\; \forall i \in \mathcal{N}$, ${\Rdi} < {\Rhi}$ and $\max_{i \in \N} \vi \leq \threshold$, the strategy $\Sh$ is a strict \NE.
\end{theorem}
\begin{proof}
    We should prove that when all nodes play strategy, $\Sh$, and an arbitrary node $i$ deviates to $\Sm$, $i$ obtains less payoff. We use overline to denote a variable if $i$ deviates.

    When everybody plays $\Sh$, $\Vh = 1$,
    and if $i$ deviates then $\overline{\Vh} = 1- \vi$. 
    One of the following two cases will occur:
    \begin{itemize}
        \item If $\vi < 1- \threshold$, $\overline{\Vh} > \threshold$;
        therefore, even if $i$ deviates, $\Ph$ executes,
        and $i$ gets $\Rdi$ which is strictly less than $\Rhi$. 
        \item If $\vi \geq 1- \threshold$, $\overline{\Vh} \leq \threshold$ and $\Ph$ does not execute with enough \power in \Sys if $i$ deviates.
        As we assumed that $\vi \leq \threshold$ and $i$ is the only player that plays $\Pm$, we will have $\overline{\Vm} = \vi < \threshold$; therefore, $\Pm$ executes with enough nodes neither and every \cn, including $i$, receives utility $0$. 
        As $\Rhi > 0$, $i$ gets less payoff if deviates.
    \end{itemize}
\end{proof}
Theorem \ref{theorem:game0} implies that in the absence of any 
external factors, given an initial honest behavior in \Sys{},
deviating from $\Ph$ has strictly less utility.
Therefore, nodes do not deviate from the honest protocol.

\subsection{Game 1: Magnate-Coordinated Collusion}
\label{subsec:game1}
We define Game 
$\Gamma_{1} = (\mathcal{N} , \{\Sh , \Smp\}^n , \utility{1}{i}{.}_{i \in \mathcal{N}})$ 
to describe \Sys in the presence of an external factor: 
the \bsc (Algorithm \ref{alg:smartcontract}). 
 Each node has two strategies $\Sh$, $\Smp$. $\Sh$ is the honest strategy as described before.
 $\Smp$ denotes the strategy of committing to the \bsc and following its commands. 
We can interpret $\Smp$ as a colluding version of $\Sm$ 
which nodes only run $\Pm$ if they are sure that enough voting \power 
is dedicated to $\Pm$.

Similarly, we denote the overall \power of players who choose $\Sh$ by $\Vh$;
furthermore, we denote the overall \power of minions (players who choose strategy $\Smp$) by $\Vmp$ with relation $\Vh + \Vmp = 1$.
Note that $\Vmp$ does not necessarily represent the real \power dedicated to $\Pm$ because if $\Vmp \leq \threshold$ then the \bsc orders minions to follow $\Ph$ and no one follows $\Pm$;
only when $\Vmp > \threshold$,
the \bsc orders minions to follow the protocol $\Pm$.
%


%
To incentivize \minions to follow the \bsc's orders unconditionally,
the \bsc requires the minions to deposit some money at the time of commitment.
%
Magnate should choose a large enough deposit 
such that it rules out any order violation.
In Theorem \ref{theorem:deposit},
we find a deposit function that satisfies this necessity.
\begin{theorem}\label{theorem:deposit}
    If the \bsc sets the deposit for all the minions as described in the equation \ref{eq:deposit}, under no circumstances any minion has the incentive to deviate from the \bsc commands.
    \begin{equation}\label{eq:deposit}
       D > \max_{i \in \N } (\Rmi + \max{\{|\Rdi| , |\Rdpi|\})}   
    \end{equation}
    
\end{theorem}
\begin{proof}
    Consider node $i$ that has committed to the \bsc and has deposited value $\Di$. $i$ receives a payoff $x$ if it follows the \bsc commands and gets a payoff $y-\Di$ if it deviates from the commands where $x, y$ are valid utility values, i.e., $x , y \in \{\Rmi,\Rhi,\Rdi,\Rdpi\}$ and their value depend on the strategy of other players.
    Our objective is to select $\Di$ in such a way that deviates from the
    commands of the \bsc are always more detrimental than any other
    strategy, regardless of what strategies other players are pursuing.
    Hence, the following should hold for any valid $x,y$:
    $$y-\Di < x \rightarrow \Di > y - x$$
    We know that as $\Rhi,\Rmi > 0$, $ (\max{\{\Rhi , \Rmi\}} + \max{\{|\Rdi| , |\Rdpi|\}}) =  \Rmi + \max{\{|\Rdi| , |\Rdpi|\}} $ is an upper bound on $y - x$; therefore, it suffices to choose $D >  \max_{i \in \N } (\Rmi + \max{\{|\Rdi| , |\Rdpi|\})}$ 
\end{proof}

The implication of Theorem \ref{theorem:deposit} is that if a rational node commits to the \bsc, it always follows the \bsc commands. Therefore there are only two possible strategies for the nodes, either playing the honest strategy or committing all of their \power to the \bsc and following its orders. 
If we use a deposit function that does not satisfy equation \ref{eq:deposit}, in some cases, some minions might benefit by deviating from the \bsc orders and dedicating less \power to the specified protocol by the \bsc even if they have committed to the \bsc.
Thus Theorem \ref{theorem:deposit} is essential for defining $\Gamma_1$. 
Now we can define the utility function of the game $\Gamma_1$ as follows:

\[
    \utility{1}{i}{s_1,\dots,s_n} =
  \begin{cases}
         {\Rhi} & \text{If } s_i= \Sh  \quad ~\& \quad \Vh > \threshold\\
          {\Rdpi}   & \text{If } s_i= \Sh \quad ~\& \quad \Vmp > \threshold\\
         {\Rhi} & \text{If } s_i= \Smp \quad \& \quad \Vh > \threshold\\
         {\Rmi} & \text{If } s_i= \Smp \quad \& \quad \Vmp > \threshold\\
         {\Rhi} & \text{If } \Vh,\Vmp \leq \threshold
        \end{cases}
  \]
   $$\text{with } \Rmi > \Rhi > 0,\Rmi > \Rdpi \text{.}$$

 The key difference between game $\Gamma_1$ and $\Gamma_0$ 
 is that the minions are now colluding and as a result, 
 they will not execute protocol $\Pm$ when $\Vh > \threshold$ to avoid
 the penalty $\Rdi$.

\begin{theorem}\label{theorem:game1}
    In the strategic-form game $\Gamma_1$, the strategy $\Sh$ is not a strict \Ne, and even further, if any subset of nodes deviates from $\Sh$ to $\Smp$, the deviating nodes always get at least the same payoff as if they were playing strategy $\Sh$.
\end{theorem}
\begin{proof}
    Without the deviation $\Vh = 1$, $\Vmp = 0$ and every node $i$ obtains reward $\Rhi$. We denote the set of nodes that deviate from $\Sh$ to $\Smp$ as $\Nm$, while the rest of the nodes $\N - \Nm$ play strategy $\Sh$. 
    We use the overlined variable to show the value of that variable if deviation takes place.
    \begin{itemize}
        \item If the overall \power of $\Nm$ is equal or less than $\threshold$, i.e., $\overline{\Vmp} \leq \threshold$,
        the \bsc will order running protocol $\Ph$; therefore, the members of $\Nm$ will run $\Ph$. 
        As other nodes also run $\Ph$, all the nodes no matter if they are a member of $\Nm$ or not will get the same reward as before, i.e., $\Rhi$.
        
        \item If the overall \power of $\Nm$ is greater than $\threshold$, i.e., $\overline{\Vmp} > \threshold$,
        the \bsc will order running protocol $\Pm$; therefore, the members of $\Nm$ will run $\Pm$ and will obtain reward $\Rmi$, and the rest of the nodes will get the utility $\Rdpi$. 
        As $\Rdpi < \Rmi$, 
        the nodes that deviate will get a better payoff, 
        and the nodes that do not deviate are better off by deviating.
        
    \end{itemize}
\end{proof}

\begin{theorem}\label{theorem:Smp-strict}
    In the strategic-form game $\Gamma_1$, if $\Rdpi < \Rmi$ and $\Rhi < \Rmi$, the strategy $\Smp$ is a strict \NE.
\end{theorem}
\begin{proof}
    When all the nodes play $\Smp$ we have $\Vmp = 1$, and every node $i$ obtains reward $\Rmi$. If player $i$ deviates to $\Sh$, one of the following two cases will occur: 
    \begin{itemize}
        \item If $\vi < 1- \threshold$, $\overline{\Vmp} = 1-\vi > \threshold$; thus, the \bsc orders to run $\Pm$ and $i$ will receive $\Rdpi < \Rmi$. 
        
        \item If $\vi \geq 1 - \threshold$, $\overline{\Vmp} = 1-\vi \leq \threshold$; thus, the \bsc orders to follow $\Ph$ 
        and every \cn, as well as $i$, 
        gets the honest reward $\Rhi < \Rmi$.
    \end{itemize}
\end{proof}

\textbf{Implication:}
In a functional system 
where \cn{}s execute the honest protocol without any 
collusion, no \cn has the incentive to deviate. 
However, with collusion,
strategy $\Sh$ becomes a weakly dominated \Ne.
%
Specifically, any colluding subset of \cns
would receive at least the same payoff as before.
Hence,
it is rational for them to deviate
in order to seek a higher payoff.
Once the subset with \power larger than $\threshold$ deviates,
the \cns strictly benefit from deviation (as $\Rmi > \Rhi$);
thus, we expect \Sys to transition to a state where everybody plays $\Smp$.
From this point, as $\Smp$ is a strict \Ne, 
no party will deviate from it.
In summary, we have identified a sequence of deviations 
where each \cn receives at least the same payoff as before, 
and eventually, the system settles into a strict \Ne and remains there.

Coming back to the example of a double-spending attack organized by a magnate, Theorem \ref{theorem:game1} states that starting from a healthy system \Sys{}, if any subset of nodes commit their \power to the \bsc and run the double-spending attack if the \bsc orders so, the minions will never get a less payoff than playing the honest strategy. Moreover, Theorem \ref{theorem:Smp-strict} suggests that starting from a situation where all the nodes commit to the \bsc and execute the double spending attack, if a node deviates and plays the honest strategy, the deviant node gets strictly less payoff after deviation.

\section{Sketch to Break Real-World Blockchain Systems}
\label{sec:realworld}

We illustrate a malicious protocol that
generally exists in real-world blockchain systems,
and then we discuss how we can use it to attack Bitcoin and Ethereum.

\begin{figure}[t!]
    \centering
    \includegraphics[scale=0.30]{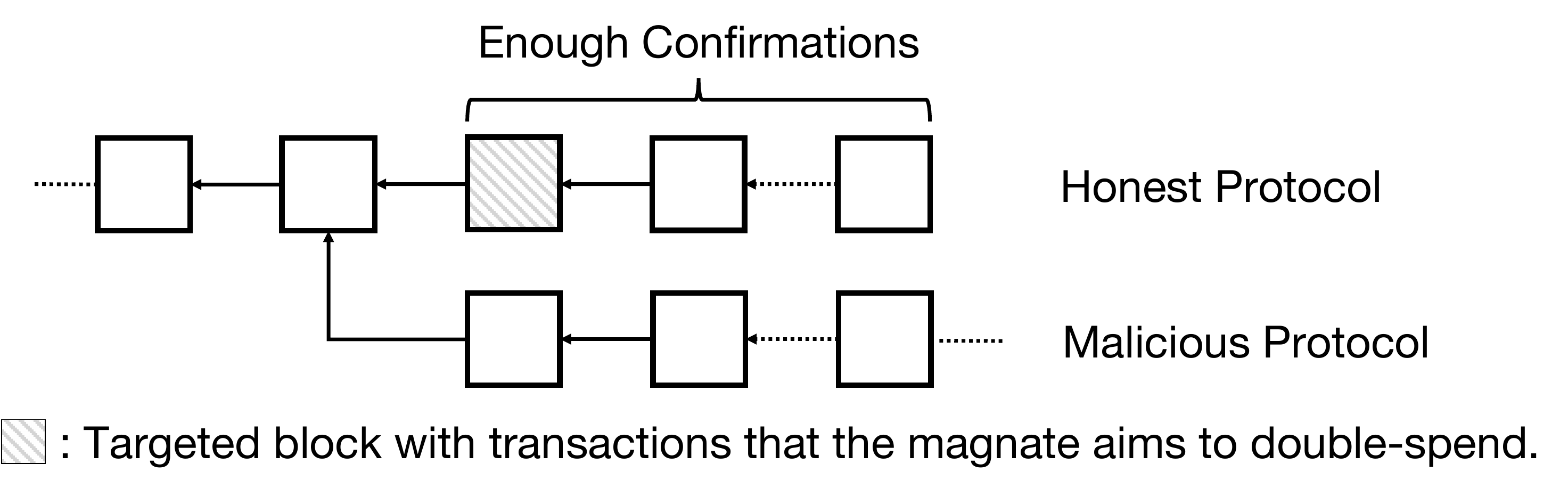}
    \caption{In a real-world blockchain system,
    given an honest protocol $\Ph$, 
    the \magnate can always construct a malicious protocol $\Pm$
    with a higher total reward by double-spending transactions 
    through reverting a confirmed block.}
    \label{fig:pmal}
\end{figure}

\subsection{Double-Spending as Malicious Protocol}
\label{subsec:pmal}






We present there always exists a malicious protocol $\Pm$
enabling double-spend attacks in \Sys{},
illustrated in Figure \ref{fig:pmal}.
A colluded \cn executes the $\Pm$
when the block that contains the target transactions receives enough block confirmations.
The protocol aims to revert the block
by working a fork, 
which allows the \magnate to double-spend the transactions
confirmed previously.
When the fork becomes the valid chain,
$\Pm$ finishes. 

\subsection{Breaking the Longest-Chain Rule in Bitcoin}
Bitcoin's protocol incentivizes the \cns to adopt the longest-chain rule
when mining a new block.
This behavior assumption applies to the rationality principle:
As long as more than 50\% of the \cns follow the longest-chain rule,
any rule-deviating \cn would reduce its expected chance
to mine new accepted blocks and thus its expected reward.
Therefore, the longest-chain rule is consistent with our Assumption 3.

We now sketch the attacking method based on double-spending.
Once a \magnate selects a transaction to double spend,
they create a \bsc with the malicious protocol
$\Pm$ in an attempt to reverse the transaction by creating a fork.
The \magnate is required to put up a deposit $\Dm$
proportional to the expected reward for double spending
this transaction.
Similar to an auction contract, the \magnate also
specifies a time $\Te$ when the contract expires.

Once the \bsc is published,
any rational \cn is incentivized to join the \bsc
and, when enough \cns have joined,
follow $\Pm$
due to the expected reward increase over following $\Ph$.
The \bsc requires \cns to deposit
$\Di$ in case they defect.
$\Di$ needs to be more than the block rewards and transactions fees
that can be reverted by the fork.
If the \bsc successfully attracts
more than 50\% of the \cns,
then the \cns launches the attack.
While launching the attack, each \cn submits
proofs to the \bsc that it is 
following $\Pm$.
Since Bitcoin uses Proof-of-Work as
the underlying consensus algorithm,
proofs can be hash results that satisfy
a difficulty requirement, similar to how
miners prove their work to a mining pool~\cite{lewenberg2015bitcoin}.

\subsection{Breaking the \DS Protocol in Ethereum}
In the recent upgrade of the Merge~\cite{ethereumMerge}, 
Ethereum changed its consensus algorithm
to \PoS.
To incentivize honest \cns and punish malicious
ones, Ethereum adopts a \ds protocol,
where each \cn must deposit some
cryptocurrency.
A \cn can withdraw its deposit entirety
when exiting the consensus group
if no other \cn can prove that it violated the protocol.
Ethereum utilizes the \ds protocol 
to punish the double-sign behavior,
\ie, a \cn signs two blocks with the same height,
thus mitigating the double-spending issues.

The \magnate can adopt a similar strategy 
to break the \ds protocol.
The \magnate still tries to double spend transactions
to create additional rewards for the colluded \cns.
The colluded \cns need to work on the fork indicated by the \magnate
after the targeted transaction is confirmed.
By doing so, 
each colluded \cn needs to sign two blocks with the same height,
a behavior violating the \ds protocol.
Thus, the colluded \cn is subject to be slashed 
if anyone submits the proof to the blockchain.
However, as long as all the colluded \cns do not allow the proof to be included
on the blockchain in the first place,
the slashing will never happen.

To prove that a \cn has executed the $\Pm$,
the \bsc has to verify
that it has voted to the fork indicated by the \magnate
and has not voted for any block with proof
potentially slashing other colluded \cns
before exiting the consensus group.
The second condition effetely delays the verification time;
However, as long as the \magnate attracts enough \cns,
the \magnate is in total control of the blockchain
before the colluded \cns exit the consensus group.
\section{Discussion}

%
Our work reveals the weakness of blockchain systems
that depend on rationality for security.
Despite this weakness,
to the best of our knowledge,
no major cryptocurrency has suffered from rational attacks~\cite{wang2021forkdec,badertscher2018but}, 
even with the usual concentration of voting power 
in the hands of a few~\cite{mariem2020all}.

The absence of such an attack may result from other factors.
First, it may be because the attack is hard to communicate and coordinate,
\ie, every \cn must be aware of such a \bsc,
rendering such attacks hard to be realized in real-world blockchain systems.
Second, cryptocurrency stakeholders may be unwilling to conduct 
such an attack due to the potential loss of faith in the 
cryptocurrency market,
leading to significant price drops;
thus, it is irrational to launch such an attack 
if we consider the monetary value of the cryptocurrency~\cite{badertscher2018but}.
Finally, some actors may choose not to participate
in such an attack out of altruism,
even though the strategy does not maximize their profits.

Nevertheless, our theoretical conclusion is that rationality
is insufficient for security;
thus, its use results in a false sense of security,
and such an attack could happen at any moment.
Our work implies that
to build a secure blockchain system,
we have to rely on non-rational assumptions, 
such as threshold assumptions
(\ie, a certain percentage of the \cns are truly honest,
even though this would lead to profit loss)
and police enforcement
(\eg, \cns would face legal prosecution if not following the honest protocol).

\section{Related Work}

The earliest work attacking blockchain rationality is selfish mining,
demonstrating that
the Bitcoin mining protocol is not incentive-compatible~\cite{eyal2018majority}. 
They prove that, in the current Bitcoin architecture, 
even if the adversary controls
less than 50\% of the hashing power, 
it can launch the attack successfully and earn more benefits than honest behavior.

Following the selfish mining attacks, 
several works attack \bc incentive mechanisms,
such as whale attacks~\cite{liao2017incentivizing},
block withholding~\cite{Eyal2015},
stubborn mining~\cite{nayak2016stubborn},
transaction withholding~\cite{Babaioff2012},
empty block mining~\cite{Houy2014}, and
fork after withholding~\cite{Know2017}.
However, these previous works only discuss the attacks in a specific protocol.




Ford et al. first outline
a general method to attack rationality,
arguing that rationality is self-defeating
when analyzing rationality in the context of a large ecosystem~\cite{ford2019rationality}.
Although the attack generally applies to any blockchain system,
it builds upon some non-existing financial instruments,
indicating the attack is not practical any time soon.
To our knowledge, 
our work is the first practical and general
attack on rationality assumptions for various blockchain systems.

Finally, utilizing smart control to incentivize malicious behaviors
is a well-known strategy in the blockchain space.
McCorry et al. present various smart contracts 
that enable bribing of miners to achieve a strategy 
that benefits the briber~\cite{mccorry2018smart}.
Juels et al. propose criminal smart contracts
that encourage the leakage of confidential information~\cite{juels2015ring}.






\section{Conclusion}

This paper proposes an attacking method that
breaks the rationality assumptions in various blockchain systems.
The attack utilizes an out-of-band smart contract to 
establish the collusion between \cns coordinated by a \magnate.
Unlike previous works
which attack rationality for a specific protocol
or rely on non-existent financial instruments,
our method is more general and practical.
Our result indicates that
the rationality assumptions do not increase the system's security
and might provide a false sense of security under the flawed analysis.
\section*{Acknowledgments}

This research was supported in part
by U.S. Office of Naval Research grant N00014-19-1-2361,
the AXA Research Fund, 
the PAIDIT project funded by ICRC, 
the IC3-Ethereum Fund,
Algorand Centres of Excellence programme managed by Algorand Foundation,
and armasuisse Science and Technology.
Any opinions, findings, and conclusions or recommendations 
expressed in this material are those of the author(s) 
and do not necessarily reflect 
the views of the funding sources.

%
%
\bibliographystyle{splncs04}
\bibliography{main.bbl}

\end{document}